\documentclass[journal,twoside,web]{ieeecolor}

\usepackage{generic}
\usepackage{cite}
\usepackage{amsmath,amssymb,amsfonts}
\usepackage{algorithmic}
\usepackage{graphicx}
\usepackage{textcomp}

\usepackage{graphicx}
\usepackage{subfigure}
\usepackage{epsfig} % for postscript graphics files
\usepackage{lcsys}
\usepackage[normalem]{ulem}
\usepackage[hidelinks]{hyperref}
\usepackage{bm}
\usepackage[most]{tcolorbox}
\usepackage{subcaption}

\newtheorem{lemma}{Lemma}
\newtheorem{theorem}{Theorem}
\newtheorem{remark}{Remark}
\newtheorem{proposition}{Proposition}

\newcommand{\E}{\mathbb{E}}
\newcommand{\KL}{D_{\rm KL}}

\newcommand{\norm}[1]{\left\|#1\right\|}
\newcommand{\normM}[2]{\left\|#1\right\|_{#2}}
\newcommand{\T}{^\top}

\def\Lc{{\mathcal L}}

\def\Nc{{\mathcal N}}

\def\Pbb{{\mathbb P}}

\def\Rbb{{\mathbb R}}

\def\0{{\bf 0}}

\newcommand{\bitem}{\begin{itemize}}
\newcommand{\eitem}{\end{itemize}}
\newcommand{\btabular}{\begin{tabular}}
\newcommand{\etabular}{\end{tabular}}
\newcommand{\bcenter}{\begin{center}}
\newcommand{\ecenter}{\end{center}}
\newcommand{\bea}{\begin{eqnarray}}
\newcommand{\eea}{\end{eqnarray}}
\newcommand{\bean}{\begin{eqnarray*}}
\newcommand{\eean}{\end{eqnarray*}}
\newcommand{\ba}{\left[ \begin{array}}
\newcommand{\ea}{\\ \end{array} \right]}
\newcommand{\bear}{\begin{array}}
\newcommand{\eear}{\\ \end{array}}

\newcommand{\non}{\nonumber}

\newcommand{\ra}{\rightarrow}

\newcounter{subequation}
\def\beasub{\addtocounter{equation}{+1}
\setcounter{subequation}{\value{equation}}
\setcounter{equation}{0}
\renewcommand{\theequation}{\arabic{subequation}\alph{equation}}
\begin{eqnarray}}
\def\eeasub{\end{eqnarray}
\setcounter{equation}{\value{subequation}}
\renewcommand{\theequation}{\arabic{equation}}}

\tcbset{phenomenonbox/.style={colback=gray!10!white,colframe=black,boxrule=0.8pt,arc=2pt,left=4pt,right=4pt,top=4pt,bottom=4pt}}

\begin{document}

\def\BibTeX{{\rm B\kern-.05em{\sc i\kern-.025em b}\kern-.08em
    T\kern-.1667em\lower.7ex\hbox{E}\kern-.125emX}}
\markboth{\journalname, VOL. XX, NO. XX, XXXX 2017}
{Author \MakeLowercase{\textit{et al.}}: Preparation of Papers for IEEE Control Systems Letters (August 2022)}

\title{Detection--Control Games under Hidden Modes: Resilience-Induced Blindness Phenomenon}
\author{Anh Tung Nguyen$^{1}$ and 
Quanyan Zhu$^{2}$
\thanks{The work of Quanyan Zhu was supported in part by the National Science Foundation under Grant No. 2521130.
The work of Anh Tung Nguyen was supported by the Swedish Foundation for Strategic Research and the Royal Swedish Academy of Engineering Sciences.}%
\thanks{$^{1}$Anh Tung Nguyen is with the Department of Information Technology, Uppsala University, Box 524, 751 20 Uppsala, Sweden; \texttt{anh.tung.nguyen@it.uu.se}.}%
\thanks{$^{2}$Quanyan Zhu is with the Department of Electrical and Computer Engineering, New York University, NY, 11201, USA; \texttt{qz494@nyu.edu}.}%
}

\maketitle
\thispagestyle{empty}

\begin{abstract}
This paper studies resilient control for cyber-physical systems operating under hidden degraded or compromised modes. We formulate hidden-mode
detection and belief-dependent control as a game between two decision makers with different objectives: the detector seeks informative belief updates, while the controller seeks regulation performance. This objective mismatch shows why the usual separation intuition between detector design and controller design may fail, leading to a performance-reversal phenomenon induced by the resilience of the controller. For a two-mode linear Gaussian system, we theoretically characterize this phenomenon by linking the resilience margin to the log-likelihood evidence. The analysis shows that a well-performing controller with a large resilience margin can suppress mode-dependent information and slow belief adaptation, which in turn degrades the control performance. The resilience-induced blindness phenomenon and its {partial} mitigation {through increased detector sensitivity} are illustrated in numerical simulations. 
\end{abstract}


\section{Introduction}
%Resilient 
\IEEEPARstart{R}{esilient} control architectures for cyber-physical systems are often designed according to a two-stage paradigm. A detector is first designed to identify anomalies, attacks, faults, or hidden operational modes from system measurements \cite{naha2023quickest,zhang2022online,ding2008model,mo2009secure,blom1988interacting}. To achieve a desired performance objective, the controller is then designed based on the detector outcome through switching, reconfiguration, or adaptive feedback \cite{tao2023optimal,wang2022asynchronous,de2023mixed}. {Related active fault-diagnosis studies retain distinct diagnosis and control tasks while coordinating their common input. The author in~\cite{xu2023input} first constructs a diagnosis-compatible input set and then selects a tracking-oriented input from that set. Guo et al. combine normalized diagnosis and control criteria in a unified finite-horizon input optimization~\cite{guo2025active}. These studies show that control actions can be coordinated to preserve diagnostic information. However, they do not characterize the inefficiency that may arise when a hidden-mode detector and a regulation-oriented controller retain separate local objectives.}

{We study this modular architecture.} The detector infers the hidden mode from measurements generated by the controlled system, while the controller uses the detector belief for regulation. Thus, control shapes the evidence available for detection, and belief adaptation shapes the control action. Treating the modules independently can consequently neglect this reciprocal information--regulation coupling. We represent this reciprocal dependence as a noncooperative game \cite{bacsar1998dynamic}. Within the stated policy classes, the detector and controller implement best responses to their respective local objectives, yet their mutual best responses need not be efficient for the overall resilient-control system.  {This inefficiency is driven by an information--regulation tradeoff. Strong regulation can suppress mode-dependent state and measurement differences and slow belief adaptation, whereas greater mode-dependent excitation can improve discrimination at the expense of regulation performance or control effort. The game framework therefore identifies and assesses this system-level inefficiency rather than creating the detector--controller architecture.}

This paper studies the interaction between a hidden-mode detector and a controller and makes three contributions. First, we formulate the interaction between the hidden-mode detector and the belief-dependent controller as a game, where the two modules find their optimal policies with respect to their local objectives. {The Nash-equilibrium viewpoint is used to explain how locally consistent module policies can be jointly inefficient. The inefficiency leads to a resilience-induced blindness phenomenon, which degrades the system performance. 
Second, since, to the best of our knowledge, this phenomenon is identified for the first time, we present it on a standard two-mode linear-Gaussian architecture with the aim of isolating the detector--controller interaction from nonlinearities and other modeling complexities.
}
For this benchmark, we derive a quantitative bound linking the belief-dependent resilience margin to log-likelihood evidence, thereby formalizing how resilient regulation can suppress detection information and degrade performance. Third, numerical simulations illustrate the phenomenon and show that increasing detector sensitivity can partially mitigate it, thereby motivating future detector--controller co-design. {The representative study exposes the mechanism clearly, while paired Monte Carlo simulations evaluate its variability across Gaussian noise realizations.}

The remainder of the paper is organized as follows. In Section~II, we introduce the game-framework interpretation of the detector--controller interaction, while Section~III characterizes the resilience-induced blindness phenomenon. Section~IV presents simulation results, and Section~V concludes the paper.

\section{{Game-Framework Interpretation of Detector--Controller Interaction}}
In this section, we aim to study the interaction between a hidden-mode detector and a controller within a standard linear Gaussian setting. The linearity will enable us to confirm a phenomenon, presented for the first time in the next section, as a direct result from the interaction. Furthermore, the Gaussian setting lays a solid foundation for future extensions motivated by the central limit theorem.

\subsection{{Detector-Controller Interaction as a Game}}
{
To study the interaction, we first describe a two-mode linear-Gaussian system as follows:
\begin{align}
x_{k+1}&=A_{\theta_k}x_k+B_{\theta_k}u_k+w_k, \label{eq:true-dynamics} \\
 y_k&=C_{\theta_k}x_k+v_k,
\label{eq:true-output}
\end{align}
}{where $\theta_k\in\Theta=\{0,1\}$ is a time-varying hidden mode. Mode $\theta_k=0$ is nominal, whereas $\theta_k=1$ is compromised or degraded. The state, input, and output are $x_k\in\Rbb^n$, $u_k\in\Rbb^m$, and $y_k\in\Rbb^p$, respectively. The disturbances are Gaussian with zero mean, i.e.,  $w_k\sim\Nc(0,\Sigma_w)$ and $v_k\sim\Nc(0,\Sigma_v)$. The matrices $A_i$, $B_i$, and $C_i$ are given for each mode index $i\in\Theta$. We assume that $\{\theta_k\}$ is piecewise constant with sparse switching. The later presented theoretical results are interval-wise: they apply on constant-mode intervals, while sparse switching determines whether enough evidence accumulates before the next switch.}

Let $\mathcal I_k^D \triangleq \sigma(y_{0:k},u_{0:k-1})$ denote the detector information and $\mathcal I_k^C \triangleq \sigma(y_{0:k},u_{0:k-1},\pi_k)$ denote the controller information. {By using $\mathcal I_k^D$ and $\mathcal I_k^C$,} the detector constructs a posterior mode belief $\pi_k\in\Delta(\Theta)$, where $\Delta(\Theta)$ is the probability simplex over $\Theta$, while the controller chooses $u_k$ using the measurement history and the current belief provided by the detector. {More specifically,} a detector policy is a sequence $\delta \triangleq \{\delta_k\}$ with $\pi_k \triangleq \delta_k(\mathcal I_k^D)$, and a controller policy is a sequence $\mu \triangleq \{\mu_k\}$ with $u_k \triangleq \mu_k(\mathcal I_k^C)$.

For a fixed controller, the detector minimizes
\begin{align}
J_D(\delta;\mu)=\E^{\delta,\mu} \textstyle \sum_{k=1}^{\infty} \big( \ell_D(\pi_k,\pi_{k-1}) + \log Z_k \big),
\label{def_det_prob}
\end{align}
where $Z_k\triangleq\sum_{i\in\Theta}\pi_{k-1}^{(i)}{\Lc_k(i)}^\beta$,
\begin{align}
\ell_D & \triangleq \KL(\pi_k\|\pi_{k-1}) - \beta \textstyle \sum_{i\in\Theta}\pi_k^{(i)}\log {\Lc_k(i)}, 
\label{eq:det_cost}
\end{align}
and the likelihood ${\Lc_k(i)}= \Pbb (y_k\mid\theta_k=i,\mathcal I_{k-1}^D)$. Here, $\KL(\pi_k\|\pi_{k-1})$ stands for the Kullback-Leibler divergence of $\pi_k$ from $\pi_{k-1}$. In \eqref{eq:det_cost}, the first term penalizes abrupt belief changes, while the second term penalizes assigning probability to modes that poorly explain the current measurement. Therefore, the detector problem \eqref{def_det_prob} is equivalent to
\begin{align}
\min_{\pi_k\in\Delta(\Theta)}\left\{\KL(\pi_k\|\pi_{k-1})-\beta \textstyle \sum_{i\in\Theta}\pi_k^{(i)}\log {\Lc_k(i)}\right\},
\label{eq:stagewise_detector}
\end{align}
which is convex on the simplex.  Using the standard Lagrangian method, \eqref{eq:stagewise_detector} gives us the following optimal solution:
\begin{align}
\pi_k^{(i)}=
\frac{\pi_{k-1}^{(i)}{\Lc_k(i)}^\beta}
{\sum_{i'\in\Theta}\pi_{k-1}^{(i')}{\Lc_k(i')}^\beta},
\qquad i\in\Theta.
\label{eq:temp_bayes_general}
\end{align}
Here, $\beta>0$ controls how strongly the detector reacts to the evidence given by the recent measurement. Interestingly, $\beta=1$ recovers the standard Bayesian update, while larger $\beta$ approaches a maximum-likelihood selector \cite[Ch. 2]{sarkka2023bayesian}. This general form \eqref{eq:temp_bayes_general} provides us with a more flexible belief update 
since altering $\beta$ changes the detector objective itself, not only the numerical value of the posterior.

\begin{figure}[!t]
    \centering
    \includegraphics[width=0.9\linewidth]{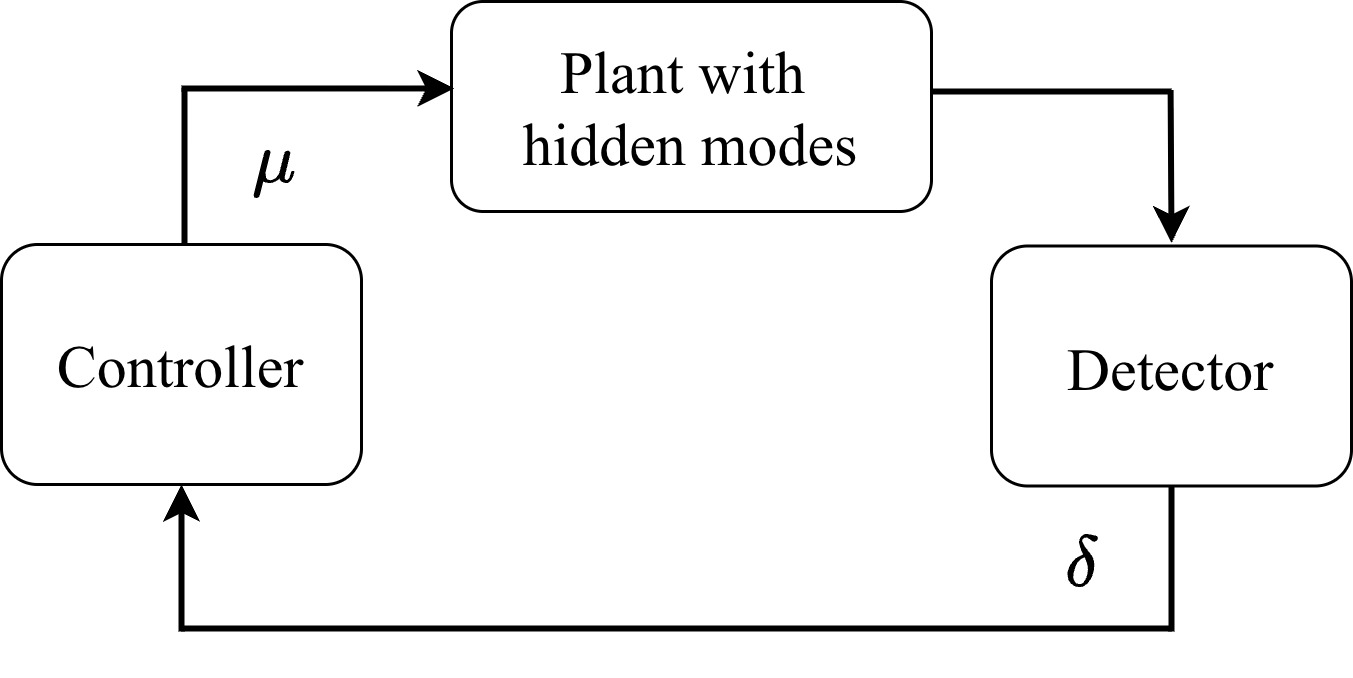}
    \caption{{An illustration of the game-theoretic interaction between the detector and the controller.} }
    \label{fig:DC_Interaction}
\end{figure}

For a fixed detector, the controller minimizes the cost:
\begin{align}
J_C(\mu;\delta)=\E^{\delta,\mu} \textstyle \sum_{k=0}^{\infty}\ell_C(x_k,u_k), \label{def_JC_general}
\end{align}
with a typical quadratic cost $\ell_C(x_k,u_k)\triangleq x_k\T Qx_k+u_k\T Ru_k~(Q \succ 0, R \succ 0)$. In general, this is a stochastic output-feedback problem since the controller only observes the state through the measurement history and the detector belief. Hence, the relevant controller information can be summarized by a state estimate, denoted as $\hat x$, and a belief $\pi_k$, i.e., $\xi_k=(\hat x_k,\pi_k)$, which motivates the belief-dependent formulation in the next section.

A pair $(\delta^\star,\mu^\star)$ is a Nash equilibrium if neither module can improve its own objective by unilateral deviation:
\begin{align}
J_D(\delta^\star;\mu^\star)&\le J_D(\delta;\mu^\star),\quad\forall\delta,\nonumber\\
J_C(\mu^\star;\delta^\star)&\le J_C(\mu;\delta^\star),\quad\forall\mu.
\label{def_br}
\end{align}
{This Nash equilibrium is locally consistent, i.e., neither the detector nor the controller can reduce its own objective by a unilateral deviation. Nevertheless, the equilibrium need not be efficient for the overall resilient-control objective. The controller objective does not internalize how regulation changes future mode information, while the detector objective does not internalize how belief adaptation changes future regulation cost. These omitted cross-effects allow individually optimal responses to yield an inefficient outcome characterized in Section III.}

{The interaction between the detector and the controller is illustrated in Fig.~\ref{fig:DC_Interaction}. 
In the remainder of the paper, we aim to study whether the best responses of these two components lead to an inefficient outcome. To facilitate the investigation, we describe the hidden-mode detector and the controller in more detail.}

\subsection{{Hidden-mode Detector}}
{For each mode index $i\in\Theta$, we assume that $(A_i,C_i)$ is observable.} Hence, a mode-conditioned state observer is given as follows:
\begin{align}
\hat x_{k|k-1}^{(i)}&=A_i\hat x_{k-1}^{(i)}+B_i u_{k-1},~
\hat y_{k|k-1}^{(i)}=C_i\hat x_{k|k-1}^{(i)},\nonumber\\
\hat x_k^{(i)}&=\hat x_{k|k-1}^{(i)}+L_i(y_k-\hat y_{k|k-1}^{(i)}),
\label{eq:mode_observer}
\end{align}
where {$(I-L_iC_i)A_i$ is Schur stable for every $i\in\Theta$.} The innovation is defined as follows:
\begin{align}
	r_k^{(i)} \triangleq y_k-\hat y_{k|k-1}^{(i)}.
    \label{eq:residual}
\end{align}
Under hypothesis $\theta_k=i$, the conditional Gaussian model gives 
\begin{align}
    y_k\mid(\theta_k=i,\mathcal I_{k-1}^D)\sim\mathcal N(\hat y_{k|k-1}^{(i)},\Sigma_i),
\end{align}
which leads to the likelihood
\begin{align}
{\Lc_k(i)}=\frac{1}{\sqrt{(2\pi)^p\det(\Sigma_i)}}
\exp\!\bigg(-\frac12 r_k^{(i)\top}\Sigma_i^{-1}r_k^{(i)}\bigg).
\end{align}

With abuse of notation, let $\pi_k$ denote the belief of the nominal mode ($\theta_k=0$) afterward. Unless otherwise stated, the detector uses the tempered Bayesian belief update \eqref{eq:temp_bayes_general} with $\beta=1$, as the standard Bayesian update in \cite[Ch. 2]{sarkka2023bayesian}: 
\begin{align}
\pi_k=\frac{\pi_{k-1} {\Lc_k(0)}}{\pi_{k-1}{\Lc_k(0)}+(1-\pi_{k-1}){\Lc_k(1)}}.
\label{eq:binary_bayes}
\end{align}
The detector-sensitivity parameter $\beta$ is revisited in Section~IV as a phenomenon mitigation policy, while the other parts stick to $\beta = 1$ to study the interaction between the detector and the controller.

{For the piecewise-constant time-varying mode, the following proposition characterizes belief concentration within a constant-mode interval. Under sparse switching, it applies locally when the dwell time is sufficiently long for evidence in favor of the active mode to accumulate.} 
\begin{proposition}
\label{prop:asym_belief}
Consider the binary-mode case $\Theta=\{0,1\}$.
Suppose that $\pi_0\in(0,1)$, ${\Lc_k(0)} > 0$, and ${\Lc_k(1)} > 0$ for all $k\geq1$. Let the detector update be given by
\eqref{eq:temp_bayes_general} with $\beta>0$. Define the accumulated log-likelihood ratio
\begin{align}
    \Lambda_k \triangleq \sum_{t=1}^{k} \log \frac{{\Lc_t(0)}}{{\Lc_t(1)}}. \label{acc_log_likelihood}
\end{align}
If $\lim_{k \ra \infty} \Lambda_k =  + \infty$,
then one has $\lim_{k \ra \infty} \pi_k = 1$.
If $\lim_{k \ra \infty} \Lambda_k =  - \infty$,
then one has $\lim_{k \ra \infty} \pi_k = 0$.
\end{proposition}

\begin{proof} 
Since ${\Lc_k(0),\Lc_k(1)}>0$ and the initial belief has full support, $\pi_k$, remains strictly positive for all $k$. Taking the ratio of the two updates yields
\begin{align*}
    \frac{\pi_k}{1-\pi_k} = \frac{\pi_{k-1}}{1-\pi_{k-1}} \left( \frac{{\Lc_k(0)}}{{\Lc_k(1)}} \right)^\beta.
\end{align*}
Define the log-belief ratio $z_k \triangleq \log \frac{\pi_k}{1-\pi_k}$, which gives us $z_k = z_{k-1} + \beta \log \frac{{\Lc_k(0)}}{{\Lc_k(1)}}$. Iterating from $0$ to $k$ gives us:
\begin{align*}
    z_k = z_0 + \beta \textstyle \sum_{t=1}^{k} \log \frac{{\Lc_t(0)}}{{\Lc_t(1)}} = z_0+\beta\Lambda_k.
\end{align*}
Note that $\pi_0 \in (0,1)$ leads to ${-\infty} < z_0 <+\infty$. If $\lim_{k \ra \infty}\Lambda_k = +\infty$, then $\lim_{k \ra \infty}z_k = +\infty$ for $\beta>0$, resulting in $\lim_{k \ra \infty}\frac{\pi_k}{1-\pi_k} = +\infty$. This implies $\lim_{k \ra \infty} \pi_k = 1$. We obtain the analogous result if $\lim_{k \ra \infty}\Lambda_k = -\infty$.
\end{proof}

Proposition~\ref{prop:asym_belief} shows that the belief update \eqref{eq:binary_bayes} can identify the true hidden mode when sufficient likelihood evidence \eqref{acc_log_likelihood} accumulates. In particular, if the accumulated log-likelihood ratio diverges in favor of one mode, then the posterior belief concentrates on that mode. In the next part, we present how the Bayesian belief update \eqref{eq:binary_bayes} can be used to design a control law.

\subsection{Belief-dependent LQR}
With the aim of showing the resilience-induced blindness phenomenon in even standard cases, we adopt the standard LQR setting for designing the controller.
For a fixed $\pi_k$ at time $k$, the objective function \eqref{def_JC_general} of the controller is converted into the following subjective LQR problem:
\begin{align}
	\min_{\{u_t\}_{t\ge k}}\sum_{t=k}^{\infty}\bar x_t\T Q\bar x_t+u_t\T Ru_t,
	\label{def_LQR_av}
\end{align}
where
\begin{align}
	\bar x_{t+1} &= \bar A_k\bar x_t+\bar B_ku_t, \non \\
	\bar A_k&=\pi_kA_0+(1-\pi_k)A_1, \non \\
	\bar B_k&=\pi_kB_0+(1-\pi_k)B_1. \non 
\end{align}
Here, we assume that the system $(\bar A_k, \bar B_k)$ is controllable for any $\pi_k \in (0,1)$.
Furthermore, the controller also considers the belief-weighted state estimate as:
\begin{align}
	\hat x_k&=\pi_k\hat x_k^{(0)}+(1-\pi_k)\hat x_k^{(1)}. \label{def_mode_est}
\end{align}

Solving \eqref{def_LQR_av} is standard using the associated Riccati equation with $P_k \succ 0$:
\begin{align}
P_k=&~Q+\bar A_k\T P_k\bar A_k
\non \\
&-\bar A_k\T P_k\bar B_k(R+\bar B_k\T P_k\bar B_k)^{-1}
\bar B_k\T P_k\bar A_k. \label{riccati_eq}
\end{align}
The belief-dependent gain and control input are
\begin{align}
K(\pi_k)&=(R+\bar B_k\T P_k\bar B_k)^{-1}\bar B_k\T P_k\bar A_k,
\non \\
u_k&=-K(\pi_k)\bar x_k.
\label{eq:lqr_gain_compact}
\end{align}
{Using a certainty-equivalent implementation for the considered belief-dependent architecture,} we can use a state-estimate feedback controller for \eqref{eq:lqr_gain_compact} as:
\begin{align}
	u_k&=-K(\pi_k)\hat x_k. \label{eq:lqr_gain_compact_hatx}
\end{align}

Next, we analyze how the belief-dependent LQR control policy \eqref{eq:lqr_gain_compact_hatx} interacts with the hidden-mode detector \eqref{eq:binary_bayes}. The interaction, resulting from the best responses of the two components \eqref{def_br}, leads to an unexpected phenomenon in the following section.
%the two hidden modes. 

\begin{remark}
{On each constant-mode interval, sufficient likelihood evidence concentrates the belief on the active mode, so the belief-weighted estimate approaches the corresponding mode-conditioned estimate. Under Schur stability of $(I-L_iC_i)A_i$ and standard Kalman-filter conditions, the estimation-error covariance is stable.} 
\end{remark}

\section{Resilience-Induced Blindness Phenomenon}
This section characterizes the resilience-induced blindness mechanism for the detector--controller equilibrium introduced in the previous section. We first derive a sufficient resilience condition for the LQR best response and then quantify how the resulting closed-loop contraction limits the likelihood evidence available to the detector. This analysis explains how mutual best responses to the two local objectives can nevertheless yield an inefficient system-level outcome.

\subsection{Belief-dependent Resilience}
Within this local interval-wise analysis, mode index $i=1$ represents the compromised dynamics. Therefore, a controller, which can stabilize the system under the two modes, is referred to as resilient to the compromised mode. 
The following lemma quantifies the resilience margin of the belief-dependent controller \eqref{eq:lqr_gain_compact} for either active mode $i\in\Theta$ on a constant-mode interval. 
\begin{lemma}[Belief-dependent LQR resilience]
\label{lem:robust_stability}
%Fix $\pi_k\in[0,1]$ and define
Consider the system \eqref{eq:true-dynamics}-\eqref{eq:true-output} under the control law \eqref{eq:lqr_gain_compact} in the noise-free case. Let us employ several auxiliary notations:
\begin{align}
\bar A_{\rm cl}&\triangleq\bar A_k-\bar B_kK(\pi_k), &
S&\triangleq Q+K(\pi_k)\T RK(\pi_k),\nonumber\\
\alpha&\triangleq\lambda_{\min}(S), &
\eta_i&\triangleq\begin{cases}1-\pi_k,&i=0,\\ \pi_k,&i=1,
\end{cases}
\end{align}
\begin{align}
\epsilon_{\rm av}^i \triangleq \eta_i\norm{A_1-A_0}+\eta_i\norm{K(\pi_k)}\norm{B_1-B_0}.
\label{lm1:epsilon_av}
\end{align}
If the following inequality holds true
\begin{align}
\alpha>2\norm{P_k}\norm{\bar A_{\rm cl}}\epsilon_{\rm av}^{i}+\norm{P_k}(\epsilon_{\rm av}^{i})^2,
\end{align}
then, there exists $m_i>0$ such that
\begin{align}
A_{\rm cl}^{i\top}P_kA_{\rm cl}^{i}-P_k\preceq -m_i I,
\label{eq:lyapunov_margin}
\end{align}
where $A_{\rm cl}^{i} \triangleq A_i-B_i K(\pi_k)$ and
\begin{align}
m_i \triangleq \alpha-2\norm{P_k}\norm{\bar A_{\rm cl}}\epsilon_{\rm av}^{i}-\norm{P_k}(\epsilon_{\rm av}^{i})^2.
\label{eq:resilience_margin}
\end{align}
Consequently, $A_{\rm cl}^{i}$ is Schur stable.
\end{lemma}
\begin{proof}
The Riccati equation \eqref{riccati_eq} gives $\bar A_{\rm cl}\T P_k\bar A_{\rm cl}-P_k=-S\preceq-\alpha I$. For the active mode index $i$, write $A_{\rm cl}^{i}=\bar A_{\rm cl}+E_i$, where $E_i=(A_i-\bar A_k)-(B_i-\bar B_k)K(\pi_k)$ and $\norm{E_i}\le\epsilon_{\rm av}^{i}$. Given $P_k$ and a Lyapunov function $V(x_k) = x_k^\top P_k x_k$, we expand its difference as follows:
\begin{align*}
&A_{\rm cl}^{i \top} P_k A_{\rm cl}^{i}-P_k = 
(\bar A_{\rm cl}+E_i)^\top
P_k (\bar A_{\rm cl}+E_i) - P_k\\
=~&
\bar A_{\rm cl}^\top P_k \bar A_{\rm cl} - P_k + \bar A_{\rm cl}^\top P_kE_i + E_i^\top P_k\bar A_{\rm cl} + E_i^\top P_kE_i.
\end{align*}
The first two terms are bounded by $-\alpha I$. For any vector $x$, the last three terms give us:
$x^\top \left( \bar A_{\rm cl}^\top P_kE_i
+ E_i^\top P_k\bar A_{\rm cl} \right)x
\le 2\|P_k\|\|\bar A_{\rm cl}\| \epsilon_{\rm av}^{i}\|x\|^2$, and
$x^\top E_i^\top P_kE_i x \le \|P_k\| \left(\epsilon_{\rm av}^{i}\right)^2 \|x\|^2$.
Consequently, by the definition \eqref{eq:resilience_margin}, one obtains \eqref{eq:lyapunov_margin}, resulting in Schur stable $A_{\rm cl}^{i}$.
\end{proof}

\subsection{Resilience-to-Detection Bound}
In this part, we leverage the resilience margin $m_i$, assumed to be positive, in Lemma~\ref{lem:robust_stability} to analyze the speed of the detector. Let us define the log-odds and log-likelihood ratio
\begin{align}
z_k \triangleq \log\frac{\pi_k}{1-\pi_k},\qquad
\ell_k \triangleq \log\frac{\Lc_k(0)}{\Lc_k(1)}.
\label{eq:log_odds_llr_compact}
\end{align}
From \eqref{eq:binary_bayes}, one has $z_k=z_{k-1}+\ell_k$. Hence, $|\ell_k|$ is the instantaneous speed of the belief update. We aim to quantify the upper bound {conditional on deterministic trajectories} for $|\ell_k|$ using $m_i$ in the following theorem.

\begin{theorem}[Resilience-to-detection bound]
\label{lem:res_det_bound}
Consider the local noise-free trajectories with gain $K(\pi_k)$ and a fixed belief $\pi_k$. Recall the resilience margin $m_i>0$ in \eqref{eq:resilience_margin} and define
\begin{align}
\rho_i= \sqrt{1-\frac{m_i}{\lambda_{\max}(P_k)}},\qquad
c_P=\sqrt{\frac{\lambda_{\max}(P_k)}{\lambda_{\min}(P_k)}}.
\end{align}
Let us assume: \\
(A1) $\Sigma_0=\Sigma_1=\Sigma\succ0$; 
\\
(A2) $\norm{x_{k}^{(i)}-\hat x_{k}^{(i)}}\le\bar e_k$ for $i\in\{0,1\}$; and
\\
(A3) $\norm{K(\pi_{k})}\le\bar K$.
\\
Then, the following inequality holds true
\begin{align}
&|z_{k+1}-z_{k}|=|\ell_{k+1}| \leq 
\bar r_{k+1}\kappa_\Sigma \omega_{k+1},
\label{eq:single_lemma_bound}
\end{align}
where $\kappa_\Sigma=\norm{\Sigma^{-1/2}}$, $\bar r_{k+1}=\max_{i\in\{0,1\}}\normM{r_{k+1}^{(i)}}{\Sigma^{-1}}$, and
\begin{align}
&\Gamma_e =  2\norm{C_1A_1} +\norm{C_1A_1-C_0A_0}+\bar K\norm{C_1B_1-C_0B_0},
\non \\
&\omega_{k+1} = \big( \norm{C_1 A_1} c_P (\rho_0 + \rho_1) \non \\
&\hspace{1.2cm}
+ \bar K c_P \max\{\rho_0, \rho_1\} \norm{C_1 B_1 - C_0 B_0} \non \\
&\hspace{1.2cm}
+ \norm{C_1 A_1 - C_0 A_0} c_P \rho_0   \big) \norm{x_{k-1}} + \Gamma_e \bar e_k.
\label{eq:Gamma_xe}
\end{align}
\end{theorem}
\begin{proof}
The proof is organized as follows: 
\begin{enumerate}
    \item[P1)] We first show that $m_i$ induces exponential decay of each mode-indexed state trajectory, yielding the upper bounds for the state, the state estimate, and the control input.
    \item[P2)] These upper bounds in P1) lead to an upper bound for $(r_{k+1}^{(0)}-r_{k+1}^{(1)})$. 
    \item[P3)] We show the relationship between $(r_{k+1}^{(0)}-r_{k+1}^{(1)})$ and $\ell(k+1)$ to obtain \eqref{eq:single_lemma_bound}.
\end{enumerate}

Let us begin with the first procedure P1) in the following.

P1) From Lemma~\ref{lem:robust_stability}, $m_i>0$ implies \eqref{eq:lyapunov_margin} and $m_i I \preceq P_k$, leading to $0 < \rho_i < 1$. Hence, in the noise-free case, \eqref{eq:lyapunov_margin} give us
\begin{align}
    \|x_{k}^{(i)}\|
\le
c_P \,\rho_i\,\|x_{k-1}\|.
\label{eq:proof_state_converge}
\end{align}
{More generally, iterating the Lyapunov inequality from time $k$ gives $\|x_{k+j}^{(i)}\|\le c_P\rho_i^j\|x_k\|$ for $j\geq 0$. Thus, exponential convergence follows from $\rho_i<1$ even when the one-step Euclidean factor $c_P\rho_i$ is not smaller than one: $c_P$ is a norm-equivalence prefactor, whereas $\rho_i$ determines the asymptotic decay rate. For the one-step residual-separation analysis below, setting $j=1$ yields \eqref{eq:proof_state_converge}. Accordingly,}
%Therefore, 
for the two mode-dependent trajectories starting from the same state $x_{k-1}$, one has
\begin{align}
\|x_{k}^{(0)}-x_{k}^{(1)}\|
\le
c_P \big(\rho_0+\rho_1\big)\|x_{k-1}\|.
\label{eq:proof_state_gap}
\end{align}
On the other hand, the assumption (A2) gives us:
\begin{align}
    \|\hat x_{k}^{(1)}-\hat x_{k}^{(0)}\|
&\le
\|x_{k}^{(1)}-x_{k}^{(0)}\|+2\bar e_k,
\label{eq:proof_state_hat_gap}
\\
\|\hat x_{k}^{(0)}\|
&\le
\|x_{k}^{(0)}\|+\bar e_k. \label{x_hat_bound}
\end{align}

Moreover, since $u_{k}=-K(\pi_{k})\hat x_{k}$ and $\|K(\pi_{k})\|\le \bar K$ in the assumption (A3), one has
\begin{align}
    \|u_{k}\|
\le~ \bar K \big( \norm{x_k} + \bar e_k \big).
\label{u_tau_bound}
\end{align}

P2) Using $r_{k+1}^{(i)}=y_{k+1}-\hat y_{k+1|k}^{(i)}$ in \eqref{eq:residual},
one has
\begin{align*}
    r_{k+1}^{(0)}-r_{k+1}^{(1)} = \hat y_{k+1|k}^{(1)}-\hat y_{k+1|k}^{(0)}.
\end{align*}
Next, the predictor \eqref{eq:mode_observer} gives us
\begin{align}
r_{k+1}^{(0)}-r_{k+1}^{(1)}
&=
C_1A_1\bigl(\hat x_{k}^{(1)}-\hat x_{k}^{(0)}\bigr)
\non \\
&\quad+
(C_1A_1-C_0A_0)\hat x_{k}^{(0)}
\non \\
&\quad+
(C_1B_1-C_0B_0)u_{k}.
\label{r_tau_bound}
\end{align}
Substituting \eqref{eq:proof_state_converge}, \eqref{eq:proof_state_gap}, \eqref{eq:proof_state_hat_gap}, \eqref{x_hat_bound}, and \eqref{u_tau_bound} into \eqref{r_tau_bound} yields
\begin{align}
&\|r_{k+1}^{(0)}-r_{k+1}^{(1)}\|
\le \omega_{k+1},
\end{align}
where $\omega_{k+1}$ and $\Gamma_e$ are defined in
\eqref{eq:Gamma_xe}. Hence, one has
\begin{align}
&d_r(k+1)
\triangleq
\|r_{k+1}^{(0)}-r_{k+1}^{(1)}\|_{\Sigma^{-1}} \leq
\kappa_\Sigma \omega_{k+1}.
\label{eq:proof_dr_bound}
\end{align}

P3) Finally, the Bayesian belief update \eqref{eq:binary_bayes} and the assumption (A1) give us:
\begin{align*}
    \ell_{k+1} = -\frac12 \left( \|r_{k+1}^{(0)}\|_{\Sigma^{-1}}^2 - \|r_{k+1}^{(1)}\|_{\Sigma^{-1}}^2 \right).
\end{align*}
Furthermore, from \eqref{eq:Gamma_xe} and the definition of $d_r(k+1)$ with its upper bound in \eqref{eq:proof_dr_bound}, one has
\begin{align}
    |\ell_{k+1}| &\le \frac12 \|r_{k+1}^{(0)}-r_{k+1}^{(1)}\|_{\Sigma^{-1}} \left( \|r_{k+1}^{(0)}\|_{\Sigma^{-1}} + \|r_{k+1}^{(1)}\|_{\Sigma^{-1}} \right) \non \\ &\le 
    \bar r_{k+1}\, d_r(k+1).
\label{eq:belief_speed_bound}
\end{align}
Combining this bound \eqref{eq:belief_speed_bound} with \eqref{eq:proof_dr_bound} and using
the fact that $z_k = z_{k-1} + \ell_k$ yields \eqref{eq:single_lemma_bound}.
\end{proof}

{Theorem~\ref{lem:res_det_bound} provides a quantitative link between resilience and detection speed by showing that the controller's best response can reduce the evidence available to the detector. Therefore, even though the controller and detector remain best responses to their respective local objectives \eqref{def_br} within the stated policy classes, the resulting mutual best responses can be inefficient for the overall resilient-control system.} Since $\rho_i \!=\!\sqrt{1-m_i/\lambda_{\max}(P_k)}$, a larger margin $m_i$ yields smaller contraction factors, which
reduce state and innovation separation and hence shrink $|\ell_{k+1}|$ with a smaller upper bound $\omega_{k+1}$. 
The cumulative evidence after a mode switch can therefore grow slowly, keeping the belief near its previous value and creating persistent gain mismatch. This explains the interaction phenomenon between the detector and the controller, which is summarized below.

\begin{tcolorbox}[phenomenonbox]
\textbf{Resilience-induced blindness phenomenon:} A large resilience margin $m_i$ yields strong contraction, reducing $d_r(k+1)$. Since $|\ell_{k+1}| \le \bar r_{k+1} d_r(k+1)$ and $z_{k+1}=z_{k}+\ell_{k+1}$, the log-odds $z_{k+1}$ can evolve slowly after a mode switch, causing persistent belief mismatch and degraded closed-loop performance.
\end{tcolorbox}

\begin{remark}
\label{rem:probability}
{Under the Gaussian estimation model, the error covariance is computable, so $\bar e_k$ in Theorem~\ref{lem:res_det_bound}-(A2) can be selected from a prescribed confidence region. Thus, Theorem~\ref{lem:res_det_bound} holds on a quantifiable high-probability event. For example, for a scalar Gaussian error, the classical three-standard-deviation rule places approximately $99.7\%$ of the probability mass within three standard deviations of the mean. Outside this event, an atypical noise realization may weaken or mask the blindness effect on an individual trajectory. Furthermore, Gaussian disturbances can approximate aggregate uncertainty from many weak sources, while the Lyapunov contraction and innovation-separation arguments underlying the analysis have nonlinear counterparts. These observations suggest extensions through local or incremental Lyapunov conditions and nonlinear likelihood bounds, whose formal development is left for future work.}
\end{remark}

\section{Numerical Simulation}
{This section illustrates the resilience-induced blindness phenomenon and evaluates its sensitivity to the detector parameter, switching period, and Gaussian noise realization. The simulation code is available at 
\href{https://github.com/tungnguyenkstnk58/Detection_Control_Game}{\textit{https://tinyurl.com/DetCtrG}}}.
We consider the two-mode linear-Gaussian system \eqref{eq:true-dynamics}--\eqref{eq:true-output} with
\begin{align}
    A_0 &= \ba{ccc}
    1.04 & 0.10 & 0 \\
    0 & 0.96 & 0.08 \\
    0 & 0 & 0.90
    \ea,
    A_1 = \ba{ccc}
    1.24 &  0.12 &   0 \\
    0 &   1.14 &   0.11 \\
    0 &   0 &   0.96
    \ea, \nonumber\\
    B_0 &= \ba{c}
    1.00 \\ 0.35 \\ 0.15
    \ea,\quad
    B_1 = \ba{c}
    0.6500 \\
    0.2275 \\
    0.0975
    \ea, \nonumber\\
    C_0 &= \ba{cc}
    1 & 0 \\
    0 & 0 \\
    0 & 1
    \ea^\top,\quad
    C_1 = \ba{cc}
    0.86 & 0 \\
    0 & 0 \\
    0 & 0.9
    \ea^\top, \nonumber\\
    \Sigma_w &=0.01I,\quad \Sigma_v=0.03I,\quad
    \Sigma_0=\Sigma_1=0.12I,\quad Q=2I. \nonumber
\end{align}
The compromised mode $i=1$ represents degraded dynamics, reduced actuation, and distorted sensing. We set $R=r_sI$ with $r_s\in\{0.25,0.5,1,2,4,8,16\}$.
Due to the fact that changing $R$ changes the LQR design objective, performance is evaluated using the common realized state cost:
\begin{align}
J_{\rm total}(r_s)\triangleq \textstyle \sum_{k=0}^{T-1}\norm{x_k}^2,
\label{eq:J_total_empirical}
\end{align}
where $T = 600$. {For each scenario, we normalize this cost by the aggressive baseline $r_s=0.25$:}
\begin{align}
\mathcal J(r_s)\triangleq100\frac{J_{\rm total}(r_s)}{J_{\rm total}(0.25)}\ \%.
\label{eq:relative_cost}
\end{align}
{Thus, $\mathcal J(r_s)>100\%$ denotes a larger realized cost than the baseline, whereas $\mathcal J(r_s)<100\%$ denotes a reversal in which a less aggressive controller outperforms the baseline.}

For each switch at $k_s$, the detection delay $T_{\rm det}(k_s)$ is the first time at which the posterior of the newly active mode exceeds $0.99$ and the aggregated delay is $T_{\rm det}^{\rm agg}=\sum_{k_s}T_{\rm det}(k_s)$. The detector is not informed of the switching times, and its posterior is not reset. Hence, one has
\begin{align}
z_{k_s+n}=z_{k_s}+\textstyle \beta\sum_{t=k_s+1}^{k_s+n}\ell_t.
\label{eq:post_switch_decomposition}
\end{align}
{Here, the inherited log-odds $z_{k_s}$ represent the detector's confidence in the previously active mode at the switching time. A strongly inherited belief requires more post-switch evidence to reverse the detector's decision. Meanwhile, the likelihood increments $\ell_t$ determine how quickly such contrary evidence is accumulated. Therefore, the detection delay depends jointly on posterior inertia and the rate of new evidence generation, with Theorem~\ref{lem:res_det_bound} characterizing the latter effect.}
\begin{figure}[!t]
\centering
\includegraphics[width=0.96\linewidth]{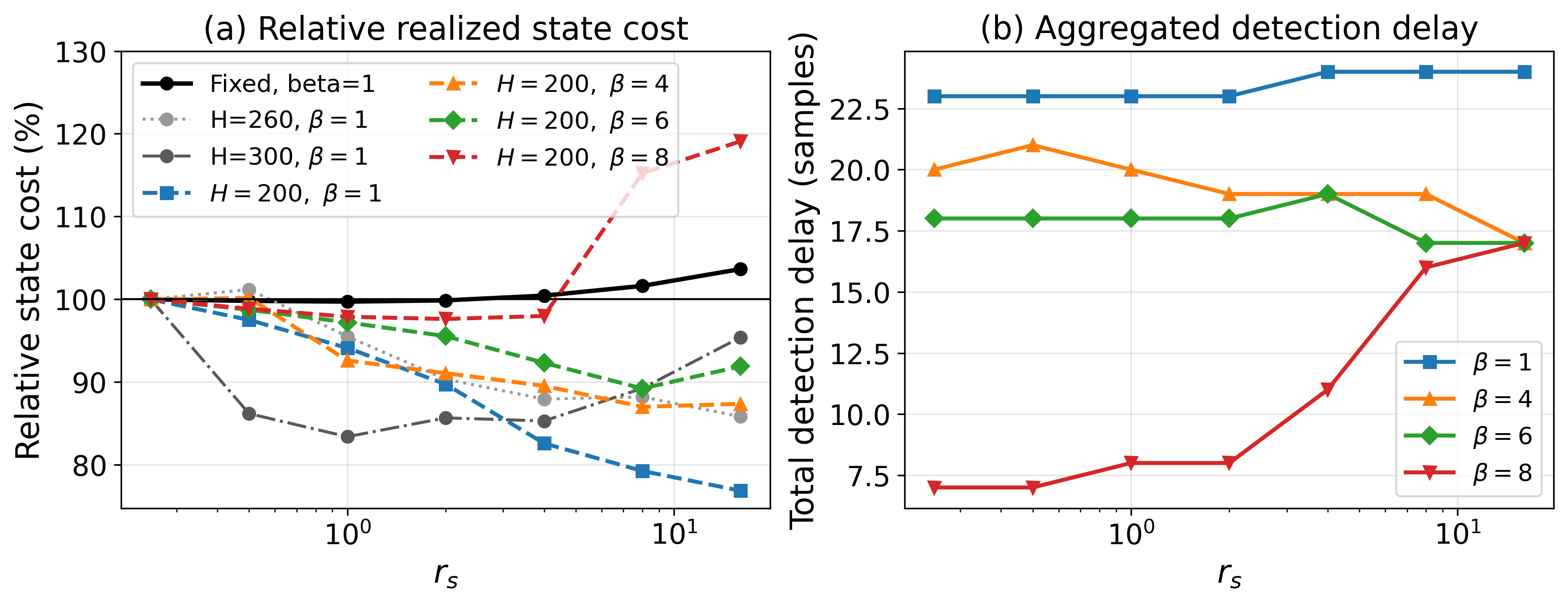}
\caption{{(a) Realized state cost relative to the $r_s=0.25$ baseline for the fixed-mode case and switching cases with different switching periods $H$ and detector sensitivities $\beta$. Values below $100\%$ reveal the resilience-induced blindness phenomenon through a reversal of the fixed-mode cost ordering. (b) Aggregated detection delay for the reference switching period $H=200$.}}
\label{fig:illustrative_relative_cost_delay}
\end{figure}

{The resilience-induced blindness phenomenon, presented in the previous section, is illustrated in
Fig.~\ref{fig:illustrative_relative_cost_delay}. Under the fixed mode, $\mathcal J(r_s)$ generally increases with $r_s$, reflecting the expected loss of regulation performance as the controller becomes less aggressive. Under mode switching with $\beta=1$, the ordering reverses over a broad range: several larger values of $r_s$ yield $\mathcal J(r_s)<100\%$. This behavior is consistent with the resilience-induced blindness phenomenon. Strong regulation suppresses state deviations but can also suppress mode-dependent evidence, thereby delaying belief adaptation and degrading performance. The switching-period curves in Fig.~\ref{fig:illustrative_relative_cost_delay}(a) show that the realized tradeoff depends on the time available for evidence accumulation between switches. Fig.~\ref{fig:illustrative_relative_cost_delay}(b) further shows that increasing $\beta$ from $1$ to $4$, $6$, or $8$ can substantially shorten the aggregated detection delay. Correspondingly, the cost curves for larger $\beta$ partially move toward the fixed-mode ordering.}

\begin{figure}[!t]
\centering
\includegraphics[width=0.96\linewidth]{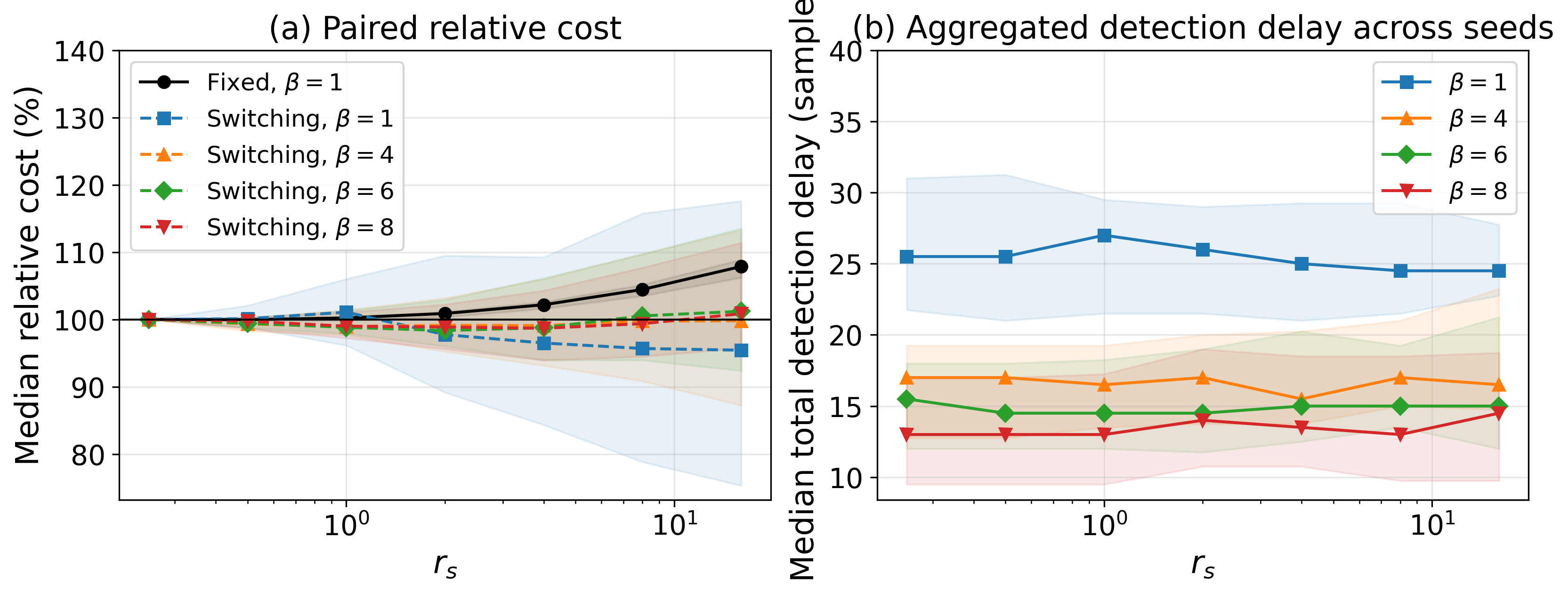}
\caption{{Paired Monte Carlo study across Gaussian noise realizations. For each realization, every $r_s$ uses the same process and measurement noise sequences and is normalized by the corresponding $r_s=0.25$ cost. The curves show the median, and the shaded regions show the interquartile range between the 25th and 75th percentiles.}}
\label{fig:monte_carlo_relative_cost_delay}
\end{figure}
{Fig.~\ref{fig:monte_carlo_relative_cost_delay} evaluates stochastic variability using paired common-noise comparisons. Since occasional large transients make the realized costs strongly skewed, we report the median and interquartile range. The median of the $\beta=1$ switching case still reflects the resilience-induced blindness phenomenon through a cost reversal over part of the tested range, although some realizations do not exhibit the reversal, consistent with Remark~\ref{rem:probability}. Increasing $\beta$ reduces the median aggregated detection delay and partially moves the switching cost toward the fixed-mode ordering, while individual Gaussian trajectories may strengthen, weaken, or mask the reversal.}

\section{Conclusion}
This paper formulated hidden-mode detection and belief-dependent control as a game between two decision makers with different objectives. {Within the stated policy classes, their locally consistent best responses can be inefficient for the overall resilient-control objective.} For two-mode linear-Gaussian systems, we characterized the resilience-induced blindness phenomenon by linking the LQR resilience margin to log-likelihood evidence. {A representative realization and paired Monte Carlo study showed that fixed-mode cost ordering can reverse under switching, although Gaussian noise may mask the reversal in individual trajectories. Future work will investigate broader equilibrium characterization, Stackelberg co-design in which one module anticipates the other's response, centralized cooperative co-design, context-dependent selection of $\beta$, and extensions to nonlinear and switching systems.
} 
\bibliographystyle{IEEEtran}
\bibliography{refs}

\end{document}